\documentclass[aps,pra,reprint,superscriptaddress,nofootinbib,longbibliography]{revtex4-2}

\usepackage{amsmath,amssymb,amsthm,mathtools}
\usepackage{bm}
\usepackage{booktabs}
\usepackage{array,tabularx}
\usepackage{graphicx}
\usepackage{microtype}
\usepackage{xcolor}
\usepackage[
  colorlinks=true,
  linkcolor=cyan!70!black,
  citecolor=cyan!70!black,
  urlcolor=cyan!70!black
]{hyperref}
\usepackage[nameinlink,capitalise]{cleveref}

\newtheorem{theorem}{Theorem}

\newtheorem{proposition}[theorem]{Proposition}
\newtheorem{lemma}[theorem]{Lemma}
\theoremstyle{definition}

\theoremstyle{remark}
\newtheorem{remark}[theorem]{Remark}

\newcommand{\R}{\mathbb{R}}
\newcommand{\C}{\mathbb{C}}
\newcommand{\Tr}{\operatorname{Tr}}
\newcommand{\diag}{\operatorname{diag}}
\newcommand{\ket}[1]{\lvert #1\rangle}
\newcommand{\bra}[1]{\langle #1\rvert}
\newcommand{\norm}[1]{\left\lVert #1\right\rVert}
\newcommand{\abs}[1]{\left\lvert #1\right\rvert}

\newcommand{\AmplitudeBound}{2.94\times 10^{-4}}
\newcommand{\AmplitudeRatioMaximum}{9.0}
\newcommand{\AmplitudeRatioMinimum}{8.1}
\newcommand{\AncillaDimension}{3\,397}
\newcommand{\CarlemanDimension}{3\,402}

\newcommand{\CoherenceAmplitude}{2.80\times 10^{-4}}
\newcommand{\CoherenceEquationResidual}{3.7\times 10^{-14}}
\newcommand{\DampingCorrection}{3.3\times 10^{-14}}
\newcommand{\DeepGeometry}{$3\times1$ periodic strip}
\newcommand{\DeepOrderRatio}{17}
\newcommand{\DissipativityMaximum}{8.7\times 10^{-19}}
\newcommand{\EncodingCondition}{1795.3}

\newcommand{\EndpointCondition}{957.2}
\newcommand{\EndpointInvariance}{4.1\times 10^{-14}}
\newcommand{\FactorizationResidual}{0}
\newcommand{\FiniteSectionResidual}{4.2\times 10^{-16}}
\newcommand{\GeneratorIdentityResidual}{3.7\times 10^{-14}}
\newcommand{\GramResidual}{4.7\times 10^{-15}}
\newcommand{\HilbertDimension}{3\,403}
\newcommand{\IsometryResidual}{2.8\times 10^{-14}}
\newcommand{\JumpCount}{3\,394}
\newcommand{\KrausCompleteness}{2.8\times 10^{-14}}
\newcommand{\KrausCount}{3\,397}
\newcommand{\LogarithmResidual}{3.2\times 10^{-14}}
\newcommand{\MetricCondition}{3.2\times 10^{6}}
\newcommand{\MetricMinimumEigenvalue}{3.1\times 10^{-7}}
\newcommand{\NegativeRealCount}{264}
\newcommand{\NeutralDimension}{9}
\newcommand{\PersistenceResidual}{7.2\times 10^{-15}}
\newcommand{\PersistenceSteps}{10}
\newcommand{\PerturbationResidual}{9.4\times 10^{-17}}
\newcommand{\PopulationDimension}{81}
\newcommand{\PowerBound}{40.68}
\newcommand{\PrimaryGeometry}{$3\times3$ periodic lattice}
\newcommand{\ReferenceInvariance}{1.7\times 10^{-16}}
\newcommand{\RelaxationTime}{0.508}
\newcommand{\RunAgreement}{1.1\times 10^{-14}}
\newcommand{\RunAmplitude}{0.020}
\newcommand{\RunDensityError}{6.5\times 10^{-5}}
\newcommand{\RunMassResidual}{3.6\times 10^{-15}}
\newcommand{\RunNonlinearError}{1.8\times 10^{-4}}
\newcommand{\RunPositivityDeviation}{1.1\times 10^{-19}}
\newcommand{\RunSteps}{10}
\newcommand{\RunTraceDeviation}{4.4\times 10^{-16}}
\newcommand{\RunVelocityError}{4.1\times 10^{-4}}
\newcommand{\SemigroupChannelResidual}{1.1\times 10^{-16}}

\newcommand{\SmallestSingularValue}{1.240\times 10^{-2}}

\newcommand{\StableDimension}{3\,393}
\newcommand{\TotalShots}{4.3\times 10^{14}}
\newcommand{\TracePreservation}{0}
\newcommand{\TransportOrthogonality}{0}

\newcommand{\UnitModulusCount}{9}

\newcommand{\AmplitudeTableBody}{%
0.005 & $2.7\times 10^{-6}$ & $1.0\times 10^{-3}$ \\
0.010 & $2.2\times 10^{-5}$ & $4.2\times 10^{-3}$ \\
0.020 & $1.8\times 10^{-4}$ & $1.7\times 10^{-2}$ \\
0.040 & $1.6\times 10^{-3}$ & $7.5\times 10^{-2}$ \\
}

\newcommand{\OrderTableBody}{%
1 & 81 & $1.3\times 10^{-3}$ \\
2 & 3\,402 & $1.8\times 10^{-4}$ \\
}

\newcommand{\DeepOrderTableBody}{%
1 & 27 & $6.9\times 10^{-4}$ \\
2 & 405 & $4.0\times 10^{-5}$ \\
3 & 4\,059 & $1.7\times 10^{-6}$ \\
}

\newcommand{\ScalingTableBody}{%
$1\times 1$ & 9 & 54 & $<0.01$ \\
$2\times 2$ & 36 & 702 & $<0.01$ \\
$3\times 3$ & 81 & 3\,402 & 0.17 \\
$4\times 4$ & 144 & 10\,584 & 1.67 \\
$8\times 8$ & 576 & 166\,752 & 414.35 \\
$16\times 16$ & 2\,304 & 2\,657\,664 & 105249.63 \\
}

\begin{document}

\title{Autonomous Lindblad Realizability of Nonunitary Linear Dynamics with a Carleman Lattice Boltzmann Application}
\author{Muhammad Idrees Khan}
\email{idreesm@chalmers.se}
\affiliation{Department of Mechanical Engineering, Chalmers University of Technology, Gothenburg, Sweden}

\begin{abstract}
    Carleman lifting converts nonlinear polynomial dynamics into finite linear
    systems, but the resulting truncations are generally nonunitary and need not
    correspond to physical quantum evolution. We prove that a finite linear
    endpoint admits an autonomous Gorini--Kossakowski--Sudarshan--Lindblad (GKSL)
    realization on vacuum coherences if and only if it is invertible and power
    bounded. The construction is explicit and realizes the nonunitary map directly
    as open-system dynamics, with no endpoint postselection and with one encoding
    and one decoding over repeated timesteps. We apply the result to the complete
    D2Q9 multiple-relaxation-time lattice Boltzmann (LB) timestep by compiling collision
    and periodic streaming into a single Carleman endpoint. The resulting GKSL
    evolution reproduces the classical Carleman trajectory over multiple timesteps,
    while the remaining discrepancy from nonlinear LB dynamics is
    the expected Carleman truncation error. The result establishes a general
    criterion for autonomous open-quantum realization of finite nonunitary
    dynamics, with Carleman--LB dynamics as a concrete example.
\end{abstract}

\maketitle
\section{Introduction}

Simulating nonlinear fluid dynamics on quantum hardware faces an immediate
structural obstruction because closed-system quantum evolution is linear and
unitary, whereas the Navier--Stokes and lattice Boltzmann (LB) equations are
not. Carleman linearization~\cite{Carleman1932,ForetsPouly2017} offers a way to
handle the nonlinearity by evolving monomials of the state instead of the state
itself. Truncating the resulting infinite hierarchy at order $K$ produces a
finite matrix $A_K$, making the dynamics linear at the cost of a modelling
error.

Carleman lifting relocates the obstruction rather than removing it. The finite
Carleman section $A_K$ is generally nonunitary and therefore cannot in general
serve directly as the evolution operator of a closed quantum system. Previous
Carleman quantum algorithms have addressed this by embedding the nonunitary
map into a larger unitary construction. Timestep-by-timestep schemes may
require measurement and renewed state preparation, while block-encoding
schemes recover the desired nonunitary action from a selected ancilla sector.
The associated postselection probability can then become a serious obstacle
to repeated evolution.

Open quantum systems offer a different route because nonunitary reduced
dynamics is itself physical. This leads to a general question about finite
linear dynamics. Given a finite linear endpoint $A$, when can it be realized
exactly by autonomous Gorini--Kossakowski--Sudarshan--Lindblad (GKSL)
dynamics~\cite{GoriniEtAl1976,Lindblad1976}? The answer depends only on the endpoint
rather than the method used to obtain it. Related Lindbladian approaches have
also been developed for nonlinear problems \cite{ChoiEtAl2026LHAM}.

The central result of this work is a characterization and explicit construction
of autonomous open-quantum realizations for finite linear dynamics. We show
that a finite real linear endpoint admits the proposed vacuum-coherence GKSL
realization exactly when it is invertible and power bounded. The criterion is
formulated entirely at the level of the finite endpoint and therefore applies
independently of the method used to obtain that endpoint. For every admissible
map, the proof constructs a time-independent Hamiltonian and Lindblad jump
operators whose GKSL evolution realizes the endpoint over repeated timesteps,
together with the corresponding finite-time completely positive trace-preserving (CPTP) channel, 
Kraus representation and Stinespring dilation.

This construction provides a direct open-system representation of nonunitary
dynamics. Rather than recovering the desired map from a selected sector of a
larger unitary evolution, the logical variables are carried by vacuum
coherences and evolve according to the required linear dynamics while the full
density operator remains physical. The same time-independent GKSL generator
can be used over successive timesteps, so the state is encoded once and
decoded once after the evolution, without intermediate measurement,
re-encoding or classical feedback.

We then specialize the general result to the complete D2Q9
multiple-relaxation-time (MRT) LB timestep. Carleman lifting converts
the nonlinear collision dynamics into a finite linear representation, while
the open-system construction addresses the nonunitarity of that representation.
Collision and periodic streaming are first composed into a single global
Carleman endpoint, and the realizability conditions are verified for the
complete collide-and-stream operator.

The result concerns physical realizability rather than implementation
complexity. In the present construction, the generator is obtained from a
dense matrix logarithm with cubic classical cost in the Carleman dimension,
and its Lindblad representation generally contains a number of nonlocal jump
operators that grows with that dimension. The final readout also remains
subject to the coherence-amplitude and sampling costs discussed below. These
costs are separate from the structural result. For admissible endpoints, 
the evolution proceeds autonomously between
successive timesteps and avoids postselection of a unitary block.

\section{Finite Carleman sections in the truncated jet algebra}
\label{sec:jet}

To make finite Carleman truncation precise, we formulate the finite section as
the action on polynomials modulo terms of degree higher than $K$. Let $F$ be a
polynomial self-map of $\R^d$ with $F(0)=0$, and let $\mathfrak{m}$ denote the
polynomials with vanishing constant term at the origin. We use the truncated
nonconstant jet algebra
\begin{equation}
J_K = \mathfrak{m}/\mathfrak{m}^{K+1},
\end{equation}
whose graded basis consists of the monomials of degrees $1$ through $K$. Let
$\Phi_K(z)$ denote the corresponding vector of monomial values. The finite
Carleman matrix $C_K(F)$ is the induced action

\begin{equation}
\Phi_K(F(z)) = C_K(F)\,\Phi_K(z) \pmod{\mathfrak{m}^{K+1}}.
\label{eq:carleman-convention}
\end{equation}
Since $F(0)=0$, substitution by $F$ preserves $\mathfrak{m}^{K+1}$, so this
finite action is well defined.

\begin{proposition}[Multiplicativity]
\label{prop:mult}
For polynomial maps $F,G$ fixing the origin,
\begin{equation}
  \begin{aligned}
  C_K(F\circ G) &= C_K(F)\,C_K(G),\\
  C_K(F^n) &= C_K(F)^n .
  \end{aligned}
  \label{eq:mult}
  \end{equation}
\end{proposition}

\begin{proof}
  Using \cref{eq:carleman-convention},
  \begin{align*}
  \Phi_K\bigl((F\circ G)(z)\bigr)
   &= C_K(F)\,\Phi_K\bigl(G(z)\bigr)\\
   &= C_K(F)\,C_K(G)\,\Phi_K(z)
   \qquad \text{in }J_K .
  \end{align*}
  Iteration gives the second identity.
\end{proof}

Two consequences organize everything that follows. First, \cref{eq:mult} is
\emph{exact} in $J_K$. It is an identity, not an approximation, and must never be
conflated with the error of the truncation against the original nonlinear map.
Those two errors are measured separately throughout \Cref{sec:results}. Second,
because $C_K(F)$ is block upper triangular with the symmetric powers of the
linear part on the diagonal,
\begin{equation}
  \begin{aligned}
  C_2(F)
   &=
   \begin{pmatrix}
   L & Q\\
   0 & L^{\odot 2}
   \end{pmatrix},\\
  F(z) &= Lz+Q(z,z).
  \end{aligned}
  \label{eq:blocks}
  \end{equation}
  the finite section inherits its spectrum from $L$. The eigenvalues of $C_K(F)$
  are the products of at most $K$ eigenvalues of $L$.

\begin{remark}[Affine maps]
\label{rem:affine}
For a general nonlinear polynomial map with a constant term,
\Cref{prop:mult} can fail because translation does not preserve
$\mathfrak{m}^{K+1}$. Appending the constant monomial gives an exact finite
representation for purely affine maps, but does not repair the general
nonlinear case under degree truncation. Two general repairs are recentering
on a fixed point of the driven problem and introducing a genuine dynamical
coordinate $s$ with
\begin{equation}
s^+ = s, \qquad f^+ = bs + Af + Q(f,f),
\end{equation}
which fixes the origin in the augmented variables, followed by restriction to the
physical slice $s=1$. Augmentation also changes the spectral structure. The extra coordinate contributes
the eigenvalue one, so if $A$ also has eigenvalue one and
$b\notin\operatorname{ran}(A-I)$, the augmented operator acquires a Jordan block
on the unit circle and leaves the hypotheses of \Cref{thm:global}.
\end{remark}

\section{The global realization theorem}
\label{sec:theorem}

\subsection{Encoding}

Fix a Hilbert space
\begin{equation}
\mathcal{H} = \operatorname{span}\{\ket{0}\}\oplus\C^{m},
\qquad h = m+1,
\end{equation}
and store the logical state in the coherences between the vacuum $\ket 0$ and the
excited block. For an embedding $J$ of full column rank, a positive definite
$\Sigma$ and a scale $\kappa>0$, set
\begin{equation}
  v(z) = \frac{\Sigma Jz}{\kappa},
\qquad
E(z) = \frac{I}{h} + \ket{v(z)}\bra{0} + \ket{0}\bra{v(z)} .
\label{eq:encoding}
\end{equation}

For a density matrix $\rho$ on $\mathcal H$, let $v_\rho\in\C^m$ denote the
excited-block coordinates of $\Pi_{\mathrm{exc}}\rho\ket0$, where
$\Pi_{\mathrm{exc}}=I-\ket0\bra0$. On the encoded subspace we use the linear
decoding map
\begin{equation}
  \operatorname{decode}(\rho)=\kappa\,(\Sigma J)^\dagger v_\rho,
\label{eq:decode}
\end{equation}
where $(\Sigma J)^\dagger$ denotes the Moore--Penrose inverse. Since $\Sigma J$ has full
column rank, $(\Sigma J)^\dagger \Sigma J=I$, and therefore
$\operatorname{decode}(E(z))=z$ exactly on the encodable set.

\begin{lemma}[Positivity budget]
\label{lem:positivity}
$E(z)$ is Hermitian with unit trace for every $\kappa$, and it is positive
semidefinite if and only if $\norm{v(z)}_2\le 1/h$. Consequently
$\kappa\ge h\sup_z\norm{\Sigma Jz}_2$ suffices on any bounded set of states.
\end{lemma}

\begin{proof}
The added term is Hermitian, off diagonal, and of rank two, supported on the
plane spanned by $\ket 0$ and $v$. On that plane the eigenvalues are
$1/h\pm\norm{v}$, and on its orthogonal complement they equal $1/h$. The trace of
the added term vanishes.
\end{proof}

\subsection{Statement}

\begin{theorem}[Global GKSL realizability]
\label{thm:global}
Let $A\in\R^{m\times m}$ and $\tau>0$. The following are equivalent.
\begin{enumerate}
\item[(i)] There exist a generator $G\in\C^{m\times m}$, an embedding $J$ of full
column rank and a Hermitian $P\succ0$ such that
\begin{equation}
e^{\tau G}J = JA,
\qquad
PG + G^{*}P \preceq 0 .
\label{eq:hyp}
\end{equation}
\item[(ii)] $A$ is invertible and power bounded, that is, $\rho(A)\le1$ and every
eigenvalue of $A$ of modulus one is semisimple.
\end{enumerate}
Moreover, when these hold there exist a time-independent Hamiltonian
$H=H^{*}$, finitely many jump operators $L_\ell$, and a scale $\kappa$ such that
the GKSL semigroup generated by $(H,\{L_\ell\})$ satisfies, for every $n\in\mathbb{N}$
and every $z$ in the encodable set,
\begin{equation}
\text{decode}\bigl(e^{n\tau\mathcal{L}}E(z)\bigr) = A^{n}z ,
\label{eq:persistence}
\end{equation}
with one encoding, one decoding, and no intervention in between. The same data
determine a Kraus family $\{K_a\}$ with $\sum_a K_a^{*}K_a=I$ realizing the same
endpoint.
\end{theorem}

\begin{proof}
\emph{(i) $\Rightarrow$ (ii).} Write $\norm{x}_P^2 = x^{*}Px$. Dissipativity
gives $\frac{d}{dt}\norm{e^{tG}x}_P^2 = (e^{tG}x)^{*}(PG+G^{*}P)(e^{tG}x)\le0$,
so $e^{tG}$ is a contraction in the metric norm for all $t\ge0$. Then
$\norm{JA^{n}z}_P = \norm{e^{n\tau G}Jz}_P\le\norm{Jz}_P$, and since $J$ is
injective with closed range the powers $A^{n}$ are uniformly bounded. A bounded
family of powers forces $\rho(A)\le1$ and excludes nontrivial Jordan blocks on
the unit circle, because such a block makes $\norm{A^{n}}$ grow
polynomially~\cite{HornJohnson2013}.
Invertibility follows because $e^{\tau G}$ is invertible and $J$ is injective.
If $JAz=0$, then $z=0$.

\emph{(ii) $\Rightarrow$ (i).} Take $J=I$. Since $A$ is invertible, it has a logarithm. We set
$G=\tau^{-1}\log A$ for any such choice. Invertibility guarantees the existence
of a logarithm while leaving the branch choice open. For the endpoints in \Cref{sec:lbm}, 
the principal branch is
unavailable because the spectrum contains eigenvalues on the closed negative
real axis, where the principal logarithm has its branch cut
(see \Cref{rem:branch}). The construction is insensitive to which branch is taken. 
Every logarithm of an eigenvalue $\lambda$ has real part $\log\abs\lambda$, 
so power boundedness places the spectrum of $G$ in the closed left half plane 
regardless of the branch choice. The imaginary-axis part is semisimple because it 
comes from the semisimple unit-modulus eigenvalues of $A$. Let $\delta>0$ separate the
strictly stable eigenvalues from the neutral ones, and take an ordered Schur 
decomposition $G=U_0TU_0^{*}$ with
\begin{equation}
  \begin{aligned}
  T &=
  \begin{pmatrix}
  T_s & X\\
  0   & T_0
  \end{pmatrix},\\
  \operatorname{Re}\operatorname{spec}(T_s) &< -\delta,
  \qquad
  \operatorname{Re}\operatorname{spec}(T_0)=0 .
  \end{aligned}
  \label{eq:schur-blocks}
  \end{equation}
  The Sylvester equation
  \begin{equation*}
  T_sY-YT_0=-X
  \end{equation*}
  has a unique solution because the two spectra are disjoint. Define
  \begin{equation*}
  \begin{aligned}
  W &=
  \begin{pmatrix}
  I & Y\\
  0 & I
  \end{pmatrix},
  &
  U &= U_0W ,
  \end{aligned}
  \end{equation*}
  so that
  \begin{equation*}
  G=U\,\diag(T_s,T_0)\,U^{-1}.
  \end{equation*}
  
  On the stable block, the Lyapunov equation
  \begin{equation*}
  T_s^{*}P_s+P_sT_s=-I
  \end{equation*}
  has a unique solution $P_s\succ0$. The neutral block is semisimple
  with purely imaginary spectrum, so
  \begin{equation*}
  T_0=V\Lambda V^{-1},
  \qquad
  P_0=(V^{-1})^{*}V^{-1}.
  \end{equation*}
  It follows that
  \begin{equation*}
  \begin{aligned}
  P_0T_0+T_0^{*}P_0
  &=(V^{-1})^{*}
    (\Lambda+\Lambda^{*})
    V^{-1}\\
  &=0.
  \end{aligned}
  \end{equation*}
  Setting
  \begin{equation*}
  P=(U^{-1})^{*}\diag(P_s,P_0)U^{-1}\succ0
  \end{equation*}
  therefore gives
  \begin{equation*}
  \begin{aligned}
  PG+G^{*}P
  &=(U^{-1})^{*}\diag(-I,0)U^{-1}\\
  &\preceq0,
  \end{aligned}
  \end{equation*}
  which is \cref{eq:hyp}.

\emph{Construction of the GKSL data.} Let $\Sigma=P^{1/2}$ and
$D=\Sigma G\Sigma^{-1}$.
Congruence by $\Sigma^{-1}$ turns dissipativity into
\begin{equation}
  \begin{aligned}
  \Gamma
   &\coloneqq -(D+D^{*}) \succeq 0,\\
  H
   &\coloneqq \frac{i}{2}(D-D^{*})=H^{*},\\
  D
   &= -iH-\frac{1}{2}\Gamma .
  \end{aligned}
  \label{eq:split}
  \end{equation}
Factor $\Gamma=\sum_\ell r_\ell^{*}r_\ell$ into row vectors and set
$L_\ell=\ket{0}r_\ell$ on $\mathcal{H}$, with the Hamiltonian $0\oplus H$. Writing
a density matrix in blocks as
$\rho=\begin{psmallmatrix}a & v^{*}\\ v & X\end{psmallmatrix}$, and using
$L_\ell\rho L_\ell^{*}\propto\ket0\bra0$ and
$\sum_\ell L_\ell^{*}L_\ell=0\oplus\Gamma$, the generator acts blockwise as
\begin{equation}
  \begin{aligned}
  \dot a &= \Tr(\Gamma X),\\
  \dot v &= Dv,\\
  \dot X &= DX+XD^{*}.
  \end{aligned}
  \label{eq:blockflow}
  \end{equation}
The middle equation is the assertion that the logical data evolve by the drift.

This is the key distinction from a unitary block encoding. The logical vector
is not required to remain a normalized pure-state amplitude vector. Its norm
may contract exactly as the nonunitary classical dynamics requires, while the
full density operator continues to evolve as a normalized positive quantum
state. Dissipation is therefore carried by the physical open-system dynamics
rather than converted into a postselection probability.

Since $e^{\tau D}\Sigma = \Sigma A$ by construction,
\cref{eq:persistence} follows for every
$n$, with $\kappa$ chosen by \Cref{lem:positivity}.

\emph{Kraus family.} Let $R=e^{\tau D}$. Dissipativity gives
\[
\frac{d}{dt}\norm{e^{tD}x}^2
=
-\bigl(e^{tD}x\bigr)^{*}\Gamma\bigl(e^{tD}x\bigr)
\le 0,
\]
hence $R^{*}R\preceq I$ and
$M\coloneqq I-R^{*}R\succeq0$ factors as $\sum_\ell s_\ell^{*}s_\ell$. Then
\begin{equation}
  \begin{aligned}
  K_0 &= \ket0\bra0\oplus R,\\
  K_\ell &= \ket0 s_\ell,\\
  K_0^{*}K_0+\sum_\ell K_\ell^{*}K_\ell
   &= \ket0\bra0\oplus(R^{*}R+M)\\
   &= I .
  \end{aligned}
  \label{eq:kraus}
  \end{equation}
and the channel maps $v\mapsto Rv$, which is what the semigroup does at time
$\tau$. The Stinespring isometry $V=\sum_a K_a\otimes\ket a$ satisfies
$V^{*}V=I$, which is \cref{eq:kraus} again. Complete positivity and the
dilation follow from the standard correspondences
~\cite{Stinespring1955,Choi1975,Kraus1983}.
\end{proof}

\subsection{Interpretation and Consequences}

\begin{remark}[Why the drift is complex]
\label{rem:complex}
The generator in \Cref{thm:global} is complex, which costs nothing because the
Hilbert space is complex. Insisting on a real generator would require a real
logarithm of a real matrix. Such a logarithm exists only when, for each negative
real eigenvalue, the Jordan blocks of every size occur an even number of times.
For the MRT parameter set used here, listed in
\Cref{app:d2q9-conventions}, all nonconserved relaxation multipliers
$1-s_k$ are negative. A uniformly applicable real construction would therefore
require additional realification machinery, potentially enlarging the state
space, with no physical benefit. Keeping the drift complex avoids this
unnecessary restriction.
\end{remark}

\begin{remark}[Logarithm branches and phases]
\label{rem:branch}
A negative endpoint eigenvalue $\lambda<0$ lies on the cut of the principal
logarithm, so it must be assigned $\log\abs\lambda+i\theta$ with $\theta$ an odd
multiple of $\pi$, and no branch makes that choice canonical. This occurs in the compiled endpoint: on the \PrimaryGeometry\ at order two, \NegativeRealCount\ of the \CarlemanDimension\ eigenvalues lie on the negative
real axis. The principal matrix logarithm is therefore undefined for this endpoint.

The realizability argument is independent of the logarithm branch because
$\operatorname{Re}\log\lambda=\log\abs\lambda$ for every branch. Different
branch choices can change the continuous generator and its
Hamiltonian--dissipative decomposition while producing the same endpoint.
Negative real endpoint eigenvalues are represented through complex phase in
the generator rather than through negative probabilities. The compiled
generator is verified by the residual
$\norm{e^{\tau G}-A}/\norm{A}$ reported in \Cref{tab:residuals}.
\end{remark}

\begin{remark}[Excluded endpoints]
\label{rem:excluded}
A singular endpoint cannot be written as $e^{\tau G}$ for any finite matrix
generator $G$. Such endpoints therefore fall outside the autonomous realization
of \Cref{thm:global}. This is not an edge case. 
As \Cref{sec:lbm} shows, it excludes $\tau_\nu=1$, a common parameter choice in LB practice. 
Endpoints with $\rho(A)>1$ are excluded because no positive metric can make the generator dissipative, 
and endpoints with a nontrivial Jordan block on the unit circle are excluded because their powers
grow without bound.
\end{remark}

\begin{remark}[One interval per timestep]
\Cref{thm:global} realizes the complete finite endpoint in one interval
$\tau$. In the LB application, collision and streaming are first composed
into the single endpoint $A_K$, so one GKSL interval corresponds to one
complete collide-and-stream timestep. The construction therefore concerns
the complete timestep rather than separate continuous-time realizations of
its collision and streaming stages.
\end{remark}

\section{Application: the complete lattice Boltzmann timestep}
\label{sec:lbm}

\subsection{The endpoint}

The model considered here is the D2Q9 LB timestep with
MRT collision and a fixed-reference-density
equilibrium~\cite{LallemandLuo2000,Succi2018}:
\begin{equation}
F_{\mathrm{LB}} = P_S\circ F_C ,
\label{eq:step}
\end{equation}
collision followed by streaming. The wrap that makes streaming a bijection of
each axis \emph{is} the periodic boundary condition, so there is no separate
boundary stage to compose. The precise population order, moment transform,
equilibrium moments, relaxation spectrum, flat lattice indexing and degree-two
monomial convention used in every numerical result are collected in
\Cref{app:d2q9-conventions}. Together, these conventions uniquely determine the compiled endpoint. 
Two structural facts make \Cref{thm:global} applicable.

\begin{proposition}[Exact quadratic closure]
\label{prop:quadratic}
At fixed reference density the polynomial equilibrium makes $F_C$ an exactly
quadratic, site-local map of the populations. Consequently, about any state
$\bar f$ with $F_{\mathrm{LB}}(\bar f)=\bar f$, the perturbation map
$\delta f\mapsto F_{\mathrm{LB}}(\bar f+\delta f)-\bar f$ is exactly quadratic,
\begin{equation}
\delta f \mapsto L\,\delta f + Q(\delta f,\delta f),
\end{equation}
with no truncation, where $L$ includes the cross term $2B(\bar f,\cdot)$.
\end{proposition}

\begin{proof}
Let $M$ be the moment transform, $m=Mf$, and $S=\diag(s_k)$ the relaxation
rates, all of which act within a single site. Their explicit values and ordering
  are given in \Cref{app:d2q9-conventions}. Each equilibrium moment is the sum
of a term linear in the density $\rho=m_0$ and the momenta $j_x=m_3$, $j_y=m_5$,
and a term quadratic in the momenta alone, the latter divided by the fixed
reference density $\rho_0$ and not by $\rho$. Write $m^{\mathrm{eq}} = \Lambda m +
\mathcal{B}_{\mathrm{eq}}(m,m)$ accordingly, with
$\mathcal{B}_{\mathrm{eq}}$ symmetric bilinear. Since $\rho$ and $j$ are
linear functionals of $f$ and $\rho_0$ is a constant, collision
\begin{equation}
  \begin{aligned}
  F_C(f)
   &= M^{-1}\bigl[(I-S)+S\Lambda\bigr]Mf\\
   &\quad +M^{-1}S\,\mathcal{B}_{\mathrm{eq}}(Mf,Mf)\\
   &= L_0f+B(f,f).
  \end{aligned}
  \label{eq:collision}
  \end{equation}
is a polynomial of degree exactly two with no constant term, and $B$ is site
local because every factor in \cref{eq:collision} is. Replacing $\rho_0$ by $\rho$ 
would make the velocity $j/\rho$ rational, so the
expansion would no longer terminate at finite order. The fixed reference
density is what gives the exact quadratic closure.
For the second claim let $P=P_S$, which is linear. Bilinearity and symmetry of
$B$ give, with no remainder,
\begin{equation}
  \begin{aligned}
  F_C(\bar f+\delta f)
   &= F_C(\bar f)\\
   &\quad+\bigl[L_0+2B(\bar f,\cdot)\bigr]\delta f\\
   &\quad+B(\delta f,\delta f).
  \end{aligned}
  \end{equation}
Apply $P$ and subtract $\bar f = PF_C(\bar f)$, which is the fixed-point
hypothesis. What remains is $L\,\delta f + Q(\delta f,\delta f)$ with
$L=P[L_0+2B(\bar f,\cdot)]$ and $Q=P\circ B$, carrying no term of degree three or
higher and sending $0$ to $0$.
\end{proof}

Expanding about a genuine fixed point is essential. The Carleman construction
requires a map that fixes the origin, and the correct first-order block is obtained 
by expanding about the fixed point. For the uniform rest reference used below, 
$\bar j_x=\bar j_y=0$, so $B(\bar f,\cdot)=0$ and the collision polynomial 
in perturbation coordinates coincides with $F_C$ itself.

\begin{proposition}[Streaming is a permutation]
\label{prop:transport}
Periodic streaming is a permutation of the population space. Each population is
carried one lattice vector along its own direction and wraps at the edge of its
axis, while its direction remains unchanged.
\end{proposition}

\begin{proof}
On each axis the shift by a fixed integer is a bijection of a cyclic group, so
the map on sites is bijective. The direction is unchanged, so the map on the
product of sites and directions is bijective as well. A bijective map of a finite
set is a permutation.
\end{proof}

Periodic streaming conserves global mass and global momentum exactly, because no
population ever changes direction. Combining
\Cref{prop:quadratic,prop:transport} with \Cref{prop:mult} gives the global
operator
\begin{equation}
A_K = C_K(F_{\mathrm{LB}}) = C_K(P_S)\,C_K(F_C),
\label{eq:global}
\end{equation}
an identity that holds exactly in the finite Carleman section.

\subsection{Verifying the hypotheses}

The hypotheses of \Cref{thm:global} are now facts about $A_K$ to be checked.

\emph{Invertibility.} Collision is invertible precisely when no relaxation
multiplier $1-s_k$ vanishes. Streaming is a permutation, hence orthogonal. 
The multipliers vanish at $\tau_\nu=1$, where the shear moments relax to
equilibrium in a single step, so that parameter choice makes $A_K$ singular and
lies outside the theorem by \Cref{rem:excluded}. We use $\tau_\nu=\RelaxationTime$. 
The full relaxation spectrum is listed in \Cref{app:d2q9-conventions}.

\emph{Power boundedness.} This condition requires a separate verification.
Streaming is a permutation, so its eigenvalues lie on the unit circle, while
the conserved moments are unrelaxed. The compiled endpoint therefore contains
neutral modes inherited from conservation and periodic transport. On the
\PrimaryGeometry\ the audit finds \UnitModulusCount\ unit-modulus eigenvalues,
all semisimple. Semisimplicity of the conserved $+1$ sector is protected structurally. 
The quadratic collision source has vanishing density and momentum moments, so
the conserved-mode functionals annihilate it and prevent Jordan coupling
between degree blocks. The complete unit-modulus sector is also
checked numerically in the compiled endpoint.

\section{Results}
\label{sec:results}

The primary numerical setting is a \PrimaryGeometry\ at order two, with
\PopulationDimension\ populations, Carleman dimension \CarlemanDimension\ and
Hilbert dimension \HilbertDimension. Accuracy measurements use geometries with
at least one axis of length three or greater, as explained in
\Cref{rem:degenerate}.

\subsection{The chain of identities}

\begin{table*}[t]
\centering
\caption{Residuals along the construction, on the \PrimaryGeometry\ at order two.
The first group reports exact algebraic identities, while the second reports
the compiled realization.}
\label{tab:residuals}
\begin{tabular}{lr}
\toprule
Quantity & Residual\\
\midrule
Stationarity of the reference state & $\ReferenceInvariance$\\
Quadratic perturbation map against the nonlinear timestep & $\PerturbationResidual$\\
Factorization $A_K=C_K(P_S)C_K(F_C)$ & $\FactorizationResidual$\\
Orthogonality of streaming, $\norm{P_S^{\mathsf T}P_S-I}$ & $\TransportOrthogonality$\\
Finite-section identity $C_K(F^n)=A_K^n$ & $\FiniteSectionResidual$\\
\midrule
Logarithm residual, $\norm{e^{\tau G}-A}/\norm{A}$ & $\LogarithmResidual$\\
Dissipativity, largest eigenvalue of $PG+G^{*}P$ & $\DissipativityMaximum$\\
Jump reconstruction $\sum_\ell r_\ell^{*}r_\ell=\Gamma$ & $\GramResidual$\\
Generator identity $D=-iH-\Gamma/2$ & $\GeneratorIdentityResidual$\\
Vacuum-coherence equation $\dot v=Dv$ & $\CoherenceEquationResidual$\\
Kraus completeness $\sum_a K_a^{*}K_a=I$ & $\KrausCompleteness$\\
Stinespring isometry $V^{*}V=I$ & $\IsometryResidual$\\
Trace preservation & $\TracePreservation$\\
Endpoint invariance $e^{\tau D}\Sigma J=\Sigma JA$ & $\EndpointInvariance$\\
Semigroup against channel & $\SemigroupChannelResidual$\\
\bottomrule
\end{tabular}
\end{table*}

\begin{figure*}[t]
\centering
\includegraphics[width=0.45\textwidth]{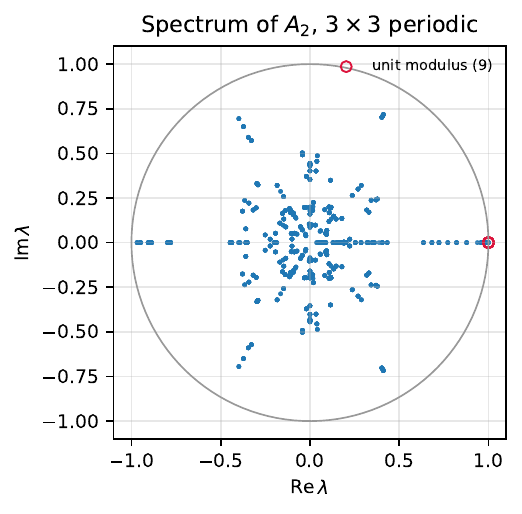}
\hfill
\includegraphics[width=0.45\textwidth]{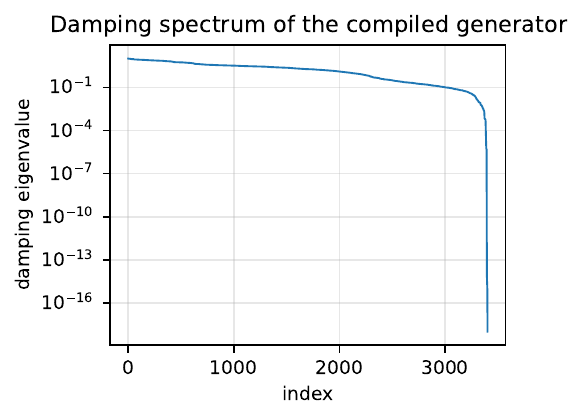}
\caption{Left: spectrum of the global endpoint $A_2$ on the \PrimaryGeometry\
against the unit circle. The \UnitModulusCount\ unit-modulus eigenvalues are neutral Carleman modes,
including conserved modes and degree-two product modes, and are all
semisimple, which is what makes the endpoint power bounded. The right panel shows the eigenvalues of the damping matrix $\Gamma$ of the compiled generator, which are nonnegative up to floating-point roundoff. The
neutral directions correspond to the unit-modulus eigenvalues shown in the left panel.}
\label{fig:spectrum}
\end{figure*}

The compiled endpoint has smallest singular value $\SmallestSingularValue$ and
condition number \EndpointCondition, spectral radius equal to $1$ to numerical
precision, and \UnitModulusCount\ unit-modulus eigenvalues, all semisimple. The
largest dyadic-power norm sampled in the numerical audit is \PowerBound.
\Cref{fig:spectrum} shows the spectrum and the
corresponding damping eigenvalues. The metric separates a stable block of dimension
\StableDimension\ from a neutral block of dimension \NeutralDimension, and has
minimum eigenvalue $\MetricMinimumEigenvalue$ and condition number
$\MetricCondition$. The realization uses \JumpCount\ jump operators and
\KrausCount\ Kraus operators, with an ancilla of dimension \AncillaDimension.
The positive-semidefinite projection applied to $\Gamma$ has norm
$\DampingCorrection$ and is reported explicitly. The endpoint channel is
independently checked through Kraus completeness.

\subsection{Autonomous evolution without intervention}

The initial population field is constructed from small random perturbations
of the density and velocity about the uniform rest reference, together with a
small nonequilibrium population perturbation so that collision is active.
Its amplitude is \RunAmplitude. The resulting positive population field is
encoded once and evolved for \RunSteps\ timesteps. The decoded state agrees with the classical
Carleman trajectory $A_K^{n}z$ to $\RunAgreement$, while the trace deviates by
$\RunTraceDeviation$ and the smallest eigenvalue of the evolved density matrix
stays above $-\RunPositivityDeviation$. An independent persistence audit at
steps $1,2,4,8,\PersistenceSteps$ gives a maximum disagreement of
$\PersistenceResidual$. The number of encodings is one and the number of
decodings is one.

Against the nonlinear MRT reference the same simulation has final
population error $\RunNonlinearError$, density error $\RunDensityError$, velocity
error $\RunVelocityError$, and mass residual $\RunMassResidual$. This error is
the modelling error of the truncation and is entirely distinct from the exact
identities of \Cref{tab:residuals}.

\subsection{Cubic scaling of the truncation error}

\begin{table}[t]
\centering
\caption{Truncation error against nonlinear MRT after \RunSteps\ timesteps on the
\PrimaryGeometry\ at order two, as a function of perturbation amplitude.}
\label{tab:amplitude}
\begin{tabular}{lrr}
\toprule
Amplitude & Absolute error & Relative error\\
\midrule
\AmplitudeTableBody
\bottomrule
\end{tabular}
\end{table}

\begin{figure}[t]
\centering
\includegraphics[width=\columnwidth]{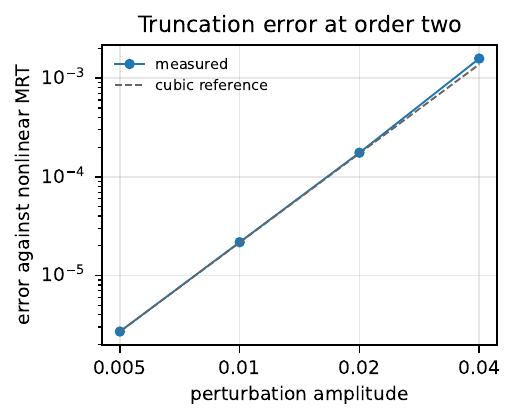}
\caption{Truncation error against the nonlinear reference after \RunSteps\
timesteps, with a cubic reference line. The measured slope follows the cubic
prediction, consistent with the leading truncation error arising from omitted
degree-three terms rather than from the GKSL realization.}
\label{fig:amplitude}
\end{figure}

Halving the amplitude reduces the error by a factor between
\AmplitudeRatioMinimum\ and \AmplitudeRatioMaximum, which is the cubic scaling
predicted by an order-two truncation whose first omitted term is of degree three,
as \Cref{fig:amplitude} shows. Raising the order improves the result, on the
primary geometry and on a smaller one where a third order is affordable.

\begin{table}[t]
\centering
\caption{Order sweeps at amplitude \RunAmplitude\ after \RunSteps\ timesteps.
Order three on the \PrimaryGeometry\ would have Carleman dimension $95\,283$,
so the third order is measured on a \DeepGeometry, which is affordable and is
not degenerate in the sense of \Cref{rem:degenerate}.}
\label{tab:order}
\begin{tabular}{lrr}
\toprule
Order $K$ & Dimension $d_K$ & Final error\\
\midrule
\multicolumn{3}{l}{\emph{\PrimaryGeometry}}\\
\OrderTableBody
\midrule
\multicolumn{3}{l}{\emph{\DeepGeometry}}\\
\DeepOrderTableBody
\bottomrule
\end{tabular}
\end{table}

Each additional order gains at least a factor of \DeepOrderRatio\ on the strip,
which is the behaviour expected of a convergent truncation and is the reason the
degenerate geometry of \Cref{rem:degenerate} must be excluded: a case that
reports zero error at order two cannot exhibit this.

\begin{remark}[Degenerate validation geometry]
\label{rem:degenerate}
A $2\times2$ periodic lattice reports zero truncation error at order two, and
the reason is lattice geometry rather than accuracy. The quadratic collision
source is even under velocity inversion, so its populations satisfy
$w_q=w_{\bar q}$. On an axis of length one or two the shifts by $c_q$ and
$-c_q$ coincide, so the local momentum of the streamed source cancels pair by
pair. Since the source re-enters the dynamics only through the local momenta,
every term of degree three and higher vanishes identically and the order-two
truncation becomes exact. The same accident afflicts every periodic lattice whose
every axis has length at most two, so the smallest lattices are exactly the ones
that cannot be used for validation. An axis of length three or more separates the
two shifts and breaks the cancellation, which is why a \PrimaryGeometry, and not
a $2\times2$ lattice, is the geometry used throughout \Cref{sec:results}.
\end{remark}

\section{Computational Cost and Limitations}
\label{sec:cost}

\begin{table*}[t]
\centering
\caption{Dimensions required at order two. The last column is the dense density
matrix in gibibytes.}
\label{tab:scaling}
\begin{tabular}{lrrr}
\toprule
Lattice & Populations & Carleman dimension & Density matrix\\
\midrule
\ScalingTableBody
\bottomrule
\end{tabular}
\end{table*}

Two costs bound the construction. Compilation is a dense classical matrix
logarithm and Schur reordering, both $O(d_K^3)$, so even the
\CarlemanDimension-dimensional endpoint used here is computationally demanding
in a dense implementation. Readout is worse in principle. \Cref{lem:positivity} caps 
the coherence amplitude at $1/h$, measured here at $\CoherenceAmplitude$ against the bound
$\AmplitudeBound$. The positivity budget implies $\abs{v_j}\le\norm{v}_2\le 1/h$ for every
coherence component. To provide a conservative resource indicator, we use the
variance-one proxy
\[
N_{\mathrm{proxy}}=\frac{1}{(\epsilon a)^2},
\]
with relative precision $\epsilon=0.01$ and amplitude scale $a=\norm{v}_2$.
For $a\sim1/h$, this proxy scales as $h^2$. Using the measured norm
$\CoherenceAmplitude$ as an optimistic amplitude scale gives approximately
$1.3\times10^{11}$ shots per component for the reported run. Multiplying this
proxy by all $d_K$ logical components gives approximately $\TotalShots$ shots.
These figures are conservative resource indicators rather than
measurement-independent lower bounds or measured hardware costs. Individual
components smaller than $\norm{v}_2$ require more shots. Decoding error is
additionally amplified by the condition number of $\Sigma J$, here
\EncodingCondition.

We therefore make no claim of quantum advantage. \Cref{thm:global} establishes
that the complete LB timestep admits a physically legitimate
open-system realization with a single encoding and a single decoding, and
identifies exactly which endpoints do not. It does not establish that evaluating
that realization is cheap, and the figures above are given so that the
distinction cannot be misread.

A second limitation is worth stating plainly. We give no analytic bound on the
error of the order-$K$ truncation against the original nonlinear dynamics.
Bounds of that kind exist for dissipative quadratic systems under a small-data
condition~\cite{ForetsPouly2017}, and connecting them to the LB setting,
where the relevant contraction is in a moment-dependent metric, is left open.
The truncation error is measured here, not bounded.

\section{Relation to existing work}

Carleman linearization has been used in quantum algorithms as a systematic
route from nonlinear classical dynamics to a larger linear problem. The
resulting finite operator is generally nonunitary, and quantum implementations
have therefore relied on unitary embeddings such as linear combinations of
unitaries and block encoding
\cite{LiuEtAl2023,SanavioSucci2024,SanavioEtAl2025}. In Carleman lattice
Boltzmann dynamics the distinction is particularly clear because streaming is
unitary whereas collision and relaxation are dissipative and nonunitary.
Earlier timestep-by-timestep Carleman LB constructions require measurement
and renewed state preparation between updates. Block-encoding constructions
instead embed the nonunitary Carleman matrix into a larger unitary, with the
desired action associated with a selected ancilla sector. The corresponding
postselection probability can make repeated evolution impractical. For
endpoints satisfying \Cref{thm:global}, the present construction realizes the
endpoint directly through open-system dynamics and therefore avoids
postselection of a unitary block.

Related work has shown that the dissipative MRT relaxation block can be
realized deterministically by a CPTP channel without block-encoding
postselection \cite{KhanSucciYao2026}. That construction treats the MRT
relaxation block itself, whereas the present work treats the complete finite
Carleman endpoint and characterizes its autonomous GKSL realizability. Recent
work has also combined Lindbladian dynamics with homotopy analysis for
nonlinear partial differential equations \cite{ChoiEtAl2026LHAM}, using
homotopy analysis to obtain an enlarged linear problem before applying open
quantum dynamics.

The present result is formulated directly at the level of a finite linear
endpoint. \Cref{thm:global} shows that the proposed autonomous GKSL realization
exists exactly when the endpoint is invertible and power bounded. The
criterion is independent of how the finite endpoint is obtained. In the
present application, Carleman lifting produces the endpoint for the complete
D2Q9 MRT collide-and-stream timestep. For an admissible endpoint, dissipation
is represented directly by the open quantum dynamics. The logical vector may 
contract while the complete density operator remains
normalized and positive. Repeated applications are generated by the same time-independent
GKSL dynamics, with autonomous evolution between the initial encoding and the
final readout.

These properties are structural consequences of the realization theorem.
The construction removes the postselection and intermediate state-preparation
requirements associated with unitary embeddings of admissible endpoints.
Questions of hardware cost and quantum speedup remain separate. The
dimensional growth of the Carleman lift, the classical cost of compiling the
generator, the complexity of realizing its jump operators and the final
readout cost remain implementation challenges. The contribution here is the
realizability result, which identifies precisely when a finite nonunitary
endpoint can be represented by the autonomous open-system construction.

\section{Conclusion}

A finite linear endpoint admits the proposed autonomous GKSL realization on
vacuum coherences exactly when it is invertible and power bounded. This gives
a necessary and sufficient matrix-level criterion rather than a
problem-specific quantum embedding. For admissible endpoints the construction
produces the Hamiltonian, the jump operators, the Kraus channel and the
Stinespring dilation associated with the same finite map.

Carleman lifting and the open-system theorem address two different structural
problems. Carleman lifting converts the nonlinear classical dynamics into a
finite linear endpoint. The GKSL construction then determines whether that
nonunitary endpoint can itself be represented by autonomous physical open
quantum dynamics. For the complete D2Q9 LB timestep, collision and periodic
streaming are compiled into one Carleman endpoint and the resulting open
evolution reproduces the classical Carleman trajectory to $\RunAgreement$
over \RunSteps\ timesteps from a single encoding. The evolution proceeds 
autonomously across timesteps from a single encoding
to the final decoding. The realizability criterion is independent of how the
finite linear endpoint is obtained. Carleman lifting provides that endpoint
in the D2Q9 application considered here, while \Cref{thm:global} applies at
the level of the endpoint itself.

The parameter choice $\tau_\nu=1$ makes the endpoint singular and therefore
falls outside the autonomous realization characterized by \Cref{thm:global},
and every periodic lattice whose axes all have length at most two makes the order-two
truncation exact by geometry, so those cases cannot be used as evidence of accuracy. 

\appendix
\section{D2Q9 conventions and explicit construction of the finite endpoint}
\label{app:d2q9-conventions}

This appendix records the conventions needed to reconstruct the lattice
Boltzmann endpoint used in the numerical results. It provides the D2Q9/MRT
data underlying \cref{eq:collision,eq:global}.

\subsection{Population order, velocities and moment transform}

At each lattice site the populations are ordered as
\begin{equation}
q=(0,1,2,3,4,5,6,7,8)
=(\mathrm{rest},\mathrm{E},\mathrm{N},\mathrm{W},\mathrm{S},
\mathrm{NE},\mathrm{NW},\mathrm{SW},\mathrm{SE}),
\end{equation}
with velocities and weights
\begin{equation}
\begin{array}{c|ccccccccc}
q&0&1&2&3&4&5&6&7&8\\
\hline
c_{qx}&0&1&0&-1&0&1&-1&-1&1\\
c_{qy}&0&0&1&0&-1&1&1&-1&-1\\
w_q&\frac49&\frac19&\frac19&\frac19&\frac19&
\frac1{36}&\frac1{36}&\frac1{36}&\frac1{36}
\end{array}.
\label{eq:appendix-d2q9}
\end{equation}

The moment vector is ordered as
\begin{equation}
m=(\rho,e,\epsilon,j_x,q_x,j_y,q_y,p_{xx},p_{xy})^{\mathsf T}=Mf,
\end{equation}
using the d'Humi\`eres--Lallemand--Luo basis~\cite{LallemandLuo2000}
\begin{equation}
M=\begin{pmatrix}
1&1&1&1&1&1&1&1&1\\
-4&-1&-1&-1&-1&2&2&2&2\\
4&-2&-2&-2&-2&1&1&1&1\\
0&1&0&-1&0&1&-1&-1&1\\
0&-2&0&2&0&1&-1&-1&1\\
0&0&1&0&-1&1&1&-1&-1\\
0&0&-2&0&2&1&1&-1&-1\\
0&1&-1&1&-1&0&0&0&0\\
0&0&0&0&0&1&-1&1&-1
\end{pmatrix}.
\label{eq:appendix-M}
\end{equation}

\subsection{Fixed-reference-density MRT collision}

The collision step is
\begin{equation}
m^+=(I-S)m+S m^{\mathrm{eq}}(m),
\qquad
f^+=M^{-1}m^+,
\label{eq:appendix-collision}
\end{equation}
with the fixed-reference-density equilibrium
\begin{equation}
m^{\mathrm{eq}}=
\begin{pmatrix}
\rho\\
-2\rho+3(j_x^2+j_y^2)/\rho_0\\
\rho-3(j_x^2+j_y^2)/\rho_0\\
j_x\\
-j_x\\
j_y\\
-j_y\\
(j_x^2-j_y^2)/\rho_0\\
j_xj_y/\rho_0
\end{pmatrix}.
\label{eq:appendix-equilibrium}
\end{equation}

Thus the equilibrium consists exactly of linear and homogeneous quadratic
terms in the moments. Its polynomial degree is therefore two. In all reported runs
\begin{equation}
\rho_0=1,
\qquad
S=\diag(0,s_e,s_\epsilon,0,s_q,0,s_q,s_\nu,s_\nu),
\end{equation}
with
\begin{equation}
  \begin{aligned}
  s_e &=1.19,
  & s_\epsilon &=1.40,\\
  s_q &=1.20,
  & \tau_\nu &=0.508=\frac{127}{250},\\
  && s_\nu &=\tau_\nu^{-1}=\frac{250}{127}.
  \end{aligned}
  \label{eq:appendix-rates}
  \end{equation}

The corresponding lattice kinematic viscosity is
\begin{equation}
\nu=\frac{\tau_\nu-\frac12}{3}
=\frac{1}{375}
\approx 2.667\times10^{-3}.
\end{equation}

The corresponding collision multipliers are
\begin{equation}
  \begin{aligned}
  1-s_k
  ={}&\biggl(
  1,-0.19,-0.40,1,-0.20,\\
  &\qquad
  1,-0.20,-\frac{123}{127},-\frac{123}{127}
  \biggr).
  \end{aligned}
  \label{eq:appendix-multipliers}
  \end{equation}

The stationary reference used for the numerical perturbation expansion is the
uniform rest equilibrium with density one,
\begin{equation}
\bar u=(0,0),
\qquad
\bar f_q=w_q
\end{equation}
at every site. Hence $F_{\mathrm{LB}}(\bar f)=\bar f$ up to the reported
floating-point stationarity residual.

\subsection{Periodic streaming and flat indexing}

For an $n_x\times n_y$ periodic lattice let
\begin{equation}
  \begin{aligned}
  \alpha(x,y,q)&=9(xn_y+y)+q,\\
  0&\le x<n_x,
  \qquad 0\le y<n_y.
  \end{aligned}
  \end{equation}
be the flat population index. Streaming sends a population at $(x,y,q)$ to
$(x+c_{qx},y+c_{qy},q)$ modulo the lattice extents, equivalently
\begin{equation}
(P_S)_{\alpha((x+c_{qx})\bmod n_x,
                     (y+c_{qy})\bmod n_y,q),\,\alpha(x,y,q)}=1,
\label{eq:appendix-streaming}
\end{equation}
with all other entries in that column zero. This specifies the permutation
matrix $P_S$ uniquely.

\subsection{Degree-two monomial convention and endpoint assembly}

For $d=9n_xn_y$ population variables, the degree-two monomials are stored once,
without $\sqrt{2}$ normalization, in lexicographic order
\begin{equation}
\delta f^{[2]}
=
\bigl(\delta f_i\delta f_j\bigr)_{0\le i\le j<d},
\qquad
\Phi_2(\delta f)=
\begin{pmatrix}
\delta f\\
\delta f^{[2]}
\end{pmatrix}.
\label{eq:appendix-lift}
\end{equation}

The symmetric power $L^{\odot2}$ is understood in this same unnormalized
monomial basis and is defined by
\begin{equation}
(Lx)^{[2]} = L^{\odot2}x^{[2]} .
\label{eq:appendix-symmetric-power}
\end{equation}

Consequently
\begin{equation}
d_2=d+\binom{d+1}{2}.
\end{equation}
For the primary $3\times3$ periodic lattice, $d=81$ and therefore
\begin{equation}
d_2=81+\binom{82}{2}=3402.
\end{equation}

Writing the exact perturbation timestep as
\begin{equation}
F_{\mathrm{LB}}(\bar f+\delta f)-\bar f
=
L\,\delta f+Q_2\,\delta f^{[2]},
\end{equation}
where $Q_2$ is obtained by collecting the homogeneous quadratic polynomial
in the $i\le j$ monomial basis. For an off-diagonal monomial
$\delta f_i\delta f_j$, the coefficient includes the two symmetric bilinear
contributions before being stored once. The order-two endpoint is
\begin{equation}
  \begin{aligned}
  A_2
   &=C_2(F_{\mathrm{LB}})\\
   &=
   \begin{pmatrix}
   L&Q_2\\
   0&L^{\odot2}
   \end{pmatrix}\\
   &=C_2(P_S)\,C_2(F_C).
  \end{aligned}
  \label{eq:appendix-A2}
  \end{equation}

Because streaming is linear,
\begin{equation}
C_2(P_S)=\diag(P_S,P_S^{\odot2}).
\end{equation}
Collision is site local, so its linear block is assembled from identical
nine-population site blocks for the uniform reference, while its quadratic
block contains only monomials whose two population indices belong to the same
site. Equations
\cref{eq:appendix-d2q9,eq:appendix-M,eq:appendix-equilibrium,eq:appendix-rates,eq:appendix-streaming,eq:appendix-lift}
therefore uniquely specify the finite endpoint used in the calculations.

\bibliography{references}

\end{document}